\documentclass[a4paper]{llncs}

\usepackage{amsmath}
\usepackage{amssymb}

\usepackage[utf8]{inputenc}% erm\"oglich die direkte Eingabe der Umlaute 
\usepackage[T1]{fontenc} % das Trennen der Umlaute

\usepackage{todonotes}
\usepackage{soul}

\usepackage{tikz}
\usetikzlibrary{automata}
\usetikzlibrary{backgrounds}
\usepackage{fancyvrb}
\usepackage{pifont}
\usepackage{stmaryrd}
\usepackage{booktabs}
\tikzset{
  initial text=,
  every path/.style={->,-stealth, every loop/.style={-stealth}},
  every initial by arrow/.style={-stealth, initial text={},
  every loop/.style={red,-stealth}},
}

\SaveVerb{sharp}|#|
\newcommand{\sharpsym}{\ensuremath{\mathop{\protect\raisebox{-0.5pt}{\protect\scalebox{1.15}{\protect\UseVerb{sharp}}}}}}
\newcommand{\errorsym}{\text{\ding{55}}}

\newcommand{\markl}{\hspace{-.5pt}\raisebox{.4pt}{\text{\scalebox{.85}{\ding{220}}}}\hspace{-.5pt}}
\newcommand{\markr}{\hspace{-.5pt}\raisebox{.4pt}{\text{\reflectbox{\scalebox{.85}{\ding{220}}}}}\hspace{-.5pt}}
\newcommand{\ssharpsym}{{\scalebox{0.7}{\sharpsym}}}

\newcommand{\copyconf}{C}
\newcommand{\newconf}{N}

\newcommand{\AP}[0]{\mathrm{AP}}

\newcommand{\psiinc}{\psi_{\text{inc}}}
\usepackage{enumerate}

\newcommand{\myquot}[1]{``#1''}

\newcommand{\nats}{\mathbb{N}}
\newcommand{\natsplus}{\nats_+}
\newcommand{\size}[1]{|#1|}
\renewcommand{\epsilon}{\varepsilon}
\renewcommand{\phi}{\varphi}
\renewcommand{\theta}[0]{\vartheta}

\newcommand{\set}[1]{\{#1\}}
\newcommand{\pow}[1]{2^{#1}}

\newcommand{\aut}{\mathcal{A}}
\newcommand{\autc}{\mathcal{T}}

\newcommand{\tm}{\mathcal M}
\newcommand{\col}{\Omega}

\newcommand{\delaygame}[1]{\Gamma\!_{f}(#1)}
\newcommand{\delaygamep}[1]{\Gamma\!_{f'}(#1)}

\newcommand{\SigmaI}{\Sigma_I}
\newcommand{\SigmaO}{\Sigma_O}
\newcommand{\strat}{\tau}
\newcommand{\stratO}{\tau_O}
\newcommand{\stratI}{\tau_I}
\newcommand{\p}{P}

\newcommand{\arena}{\mathcal{A}}

\newcommand{\exptime}{\textsc{ExpTime}}

\newcommand{\bigo}{\mathcal{O}}

\newcommand{\pspace}{{\textsc{PSpace}}}

\newcommand{\twoexp}{{\textsc{2ExpTime}}}
\newcommand{\threeexp}{{\textsc{3ExpTime}}}
\newcommand{\atwoexpspace}{\textsc{A2ExpSpace}}

\newcommand{\autp}{\mathcal{P}}
\newcommand{\resolve}{r}
\newcommand{\game}{\mathcal{G}}

\newcommand{\curlyR}{\mathfrak{R}}
\newcommand{\dom}{\mathrm{dom}}

\newcommand{\block}[1]{\overline{#1}}

\newcommand{\qacc}{q_{A}}
\newcommand{\qrej}{q_{R}}

\newcommand{\cceq}{\mathop{::=}}

\newcommand{\update}{\mathrm{upd}}

\newcommand{\F}{{\mathbf{F\,}}}
\newcommand{\G}{{\mathbf{G\,}}}
\newcommand{\Fp}{{\mathbf{F_{\!P}\,}}}
\newcommand{\U}{{\mathbf{\,U\,}}}
\newcommand{\X}{{\mathbf{X\,}}}
\newcommand{\R}{{\mathbf{\,R\,}}}

\newcommand{\rel}[0]{\mathrm{rel}}

\newcommand{\ltl}{\ensuremath{\textsc{LTL}}}
\newcommand{\pltl}{\mathrm{PLTL}}
\newcommand{\prompt}{\ensuremath{\textsc{Prompt}\text{-}}\ltl}

\newcommand{\pldl}{\textsc{PLDL}}
\newcommand{\ldl}{\textsc{LDL}}

\title{Prompt Delay\thanks{The first author was supported by an IMPRS-CS PhD Scholarship, the second by the project \myquot{TriCS} (ZI~1516/1-1) of the German Research Foundation (DFG). }}

\author{Felix Klein and Martin Zimmermann}

\institute{Reactive Systems Group, Saarland University, Germany\\
 \email{\{klein, zimmermann\}@react.uni-saarland.de}}

\begin{document}

\maketitle

%%%%%%%%%%%%%%%%%%%%%%%%%%%%%%%%%%%%%%%%%%%%%%%%%%%%%%%%%%%%
%%%%%%%%%%%%%%%%%%%%%%%%%%%%%%%%%%%%%%%%%%%%%%%%%%%%%%%%%%%%
\begin{abstract}
  Delay games are two-player games of infinite duration in which one
player may delay her moves to obtain a lookahead on her opponent's
moves. Recently, such games with quantitative winning conditions in
weak MSO with the unbounding quantifier were studied, but their properties turned out
to be unsatisfactory. In particular, unbounded lookahead is in general
necessary.

\sloppy

Here, we study delay games with winning conditions given by {\prompt},
Linear Temporal Logic equipped with a parameterized eventually
operator whose scope is bounded. Our main result shows that solving
{\prompt} delay games is complete for triply-exponential
time. Furthermore, we give tight triply-exponential bounds on the
necessary lookahead and on the scope of the parameterized eventually
operator. Thus, we identify {\prompt} as the first known class of
well-behaved quantitative winning conditions for delay games.

\fussy

Finally, we show that applying our techniques to delay games with 
$\omega$-regular winning conditions answers open questions in the 
cases where the winning conditions are given by non-deterministic, 
universal, or alternating automata.

\end{abstract}

%%%%%%%%%%%%%%%%%%%%%%%%%%%%%%%%%%%%%%%%%%%%%%%%%%%%%%%%%%%%
%%%%%%%%%%%%%%%%%%%%%%%%%%%%%%%%%%%%%%%%%%%%%%%%%%%%%%%%%%%%
\section{Introduction}
\label{sec_intro}
The synthesis of reactive systems concerns the automatic construction of
an implementation satisfying a given specification against every behavior of its possibly antagonistic environment. A prominent
specification language is Linear Temporal Logic ({\ltl}), describing
the temporal behavior of an
implementation~\cite{Pnueli77}. The {\ltl} synthesis problem
has been intensively studied since the seminal work of Pnueli and Rosner~\cite{DBLP:conf/popl/PnueliR89,DBLP:conf/icalp/PnueliR89}, theoretical foundations have been established~\cite{DBLP:journals/tocl/AlurT04,DBLP:conf/focs/KupfermanV05}, and
several tools have been developed~\cite{DBLP:journals/fmsd/Ehlers12,DBLP:journals/fmsd/FiliotJR11,DBLP:journals/sttt/FinkbeinerS13}.

However, {\ltl} is not able to express quantitative properties. As an
example, consider the classical request-response
condition~\cite{DBLP:journals/ita/HornTW015}, where every
request~$ q $ has to be answered eventually by some response~$ r $. This property is expressible in {\ltl} via the
formula~$ \G (q \rightarrow \F r) $, but the property cannot guarantee
any bound on the waiting times between a request and its earliest response.
To specify such a behaviour, parameterized logics have been
introduced~\cite{DBLP:journals/tocl/AlurETP01,FaymonvilleZimmermann14,KupfermanPitermanVardi09,Zimmermann15c}, which extend $\ltl$ by quantitative operators. %This research is part of the paradigm shift from qualitative synthesis to quantitative synthesis.

The simplest of these logics is 
{\prompt}~\cite{KupfermanPitermanVardi09}, which extends
{\ltl} by the prompt eventually operator~$ \Fp $\!.\!\! The scope of this operator is bounded by some arbitrary but fixed number~$k$. With this extension, we can express the aforementioned property by the {\prompt}
formula~$ \G (q \rightarrow \Fp r) $, expressing that every request is answered within $k$ steps.  To show that the {\prompt} synthesis problem is as hard as the {\ltl} synthesis problem, i.e., $\twoexp$-complete, Kupferman et al.\ introduced the alternating-color technique to reduce the former problem to the latter~\cite{KupfermanPitermanVardi09}. Additionally, similar
reductions have been proven to exist in other settings too,
where {\prompt} can be reduced to {\ltl} using the alternating-color
technique, e.g., for (assume-guarantee) model-checking~\cite{KupfermanPitermanVardi09}. Finally, the technique is also applicable to more expressive extensions of $\ltl$, e.g., parametric {\ltl}~\cite{Zimmermann13}, parametric $\ldl$~\cite{FaymonvilleZimmermann14}, and their variants with costs~\cite{Zimmermann15c}.

Nevertheless, all these considerations assume that the specified
implementation immediately reacts to inputs of the environment. However, this
assumption might be too restrictive, e.g., in a buffered
network, where the implementation may delay its outputs by several time
steps. Delay games have been introduced by Hosch and
Landweber~\cite{HoschLandweber72} to overcome this restriction.  In the setting of
infinite games, the synthesis problem is viewed as a game
between two players, the input player ``Player~$ I $'', representing the environment, and the output
player~``Player~$ O $'', representing the implementation. The goal of Player~$ O $ is to satisfy the
specification, while Player~$ I $ tries to violate it. Usually, the players move in strict alternation. On the contrary, in a delay
game, Player~$O$ can delay her moves to obtain a
lookahead on her opponent's moves. This way, she gains additional
information on her opponent's strategy, which she can use to achieve her
goal. Hence, many specifications are realizable, when allowing
lookahead, which are unrealizable otherwise. 

For delay games with $\omega$-regular winning conditions (given by deterministic parity automata) exponential lookahead is always sufficient and in general necessary, and determining the winner is $\exptime$-complete~\cite{KleinZimmermann16a}. As $\ltl$ formulas can be translated into equivalent deterministic parity automata of doubly-exponential size, these results imply a triply-exponential upper bound on the necessary lookahead in delay games with $\ltl$ winning condition and yield an algorithm solving such games with triply-exponential running time. However, no matching lower bounds are known.

Recently, based on the techniques developed for the $\omega$-regular case, the investigation of delay games with quantitative winning conditions
was initiated by studying games with winning conditions specified in weak monadic second order logic with the unbounding quantifier (WMSO$+$U)~\cite{DBLP:conf/csl/Bojanczyk04}. This logic
extends the weak variant of monadic second order logic
(WMSO), where only quantification over
finite sets is allowed, with an additional unbounding
quantifier that allows to express (un)boundedness properties. The resulting logic subsumes all parameterized logics mentioned above. The winner of a WMSO$+$U delay game with respect to bounded lookahead is effectively computable~\cite{Zimmermann16}. However, in general, Player~$O$ needs unbounded lookahead to win such games and the decidability of such games with respect to arbitrary lookahead remains an open problem. In the former aspect, delay games with WMSO$+$U winning conditions behave worse than those with $\omega$-regular ones.

\medskip

\noindent\textbf{Our Contribution.} The results on WMSO$+$U delay games show the relevance of exploring more restricted classes of quantitative winning conditions which are better-behaved. In particular, bounded lookahead should always suffice and the winner should be effectively computable.
To this end, we investigate delay games with {\prompt} winning
conditions. Formally, we consider the following synthesis problem:
given some {\prompt} formula~$ \varphi $, does there exist some
lookahead and some bound~$ k $ such that Player~$ O $ has a
strategy producing only
outcomes that satisfy $ \varphi $ with respect to the bound~$k$ (and, if yes, compute such a strategy)?

We present the first results for
delay games with {\prompt} winning conditions. First, we show that the synthesis
problem is in $\threeexp$ by tailoring the alternating-color technique to delay games and integrating it into the algorithm developed for the $\omega$-regular case. In the end, we obtain a reduction from delay games with {\prompt} winning conditions to delay-free parity games of triply-exponential size. 

Second, from this construction, we derive 
triply-exponential upper bounds on the necessary lookahead 
for Player~$ O $, i.e., bounded lookahead always suffices, as well as a triply-exponential upper bound on the necessary scope of the prompt eventually operator. Thus, we obtain  the same upper bounds as for $\ltl$.

Third, we complement all three upper bounds by matching lower bounds, e.g., the problem is $\threeexp$-complete and there are triply-exponential lower bounds on the necessary lookahead and on the scope of the prompt eventually operator.  The former two lower bounds already hold for the special case of $\ltl$ delay games. Thereby, we settle the case of delay games with $\ltl$ winning conditions as well as the case of delay games with $\prompt$ winning conditions and show that they are of equal complexity and that the same bounds on the necessary lookahead hold. Thus, we prove that delay games with $\prompt$ winning conditions are not harder than those with $\ltl$ winning conditions. The complexity of solving {\ltl} games increases exponentially when adding lookahead, which is in line with the results in the $\omega$-regular case~\cite{KleinZimmermann16a}, where one also observes an exponential blowup.

Fourth, our proofs are all applicable to the stronger extensions of $\ltl$ like parametric {\ltl}~\cite{Zimmermann13}, parametric $\ldl$~\cite{FaymonvilleZimmermann14}, and their variants with costs~\cite{Zimmermann15c}, as the alternating-color technique is applicable to them as well and as their formulas can be compiled into equivalent exponential Büchi automata. 

Fifth, we show that our lower bounds also answer open questions in the $\omega$-regular case mentioned in~\cite{KleinZimmermann16a}, e.g., on the influence of the branching mode of the specification automaton on the complexity. Recall that the tight exponential bounds on the complexity and the necessary lookahead for $\omega$-regular delay games were shown for winning conditions given by \emph{deterministic} automata. Our lower bounds proven here can be adapted to show that both these bounds are doubly-exponential for non-deterministic and universal automata and triply-exponential for alternating automata. Hence, the lower bounds match the trivial upper bounds obtained by determinizing the automata and applying the results from~\cite{KleinZimmermann16a}. Thus, we complete the picture in the $\omega$-regular case with regard to the branching mode of the specification automaton. 

\medskip

\sloppy

\noindent\textbf{Related Work.} Delay games with $ \omega $-regular
winning conditions have been introduced
by Hosch and Landweber, who proved that the winner w.r.t.\ bounded lookahead can be determined effectively~\cite{HoschLandweber72}. Later, they were revisited by Holtmann et al.\ who showed that bounded lookahead is always sufficient and who gave a streamlined algorithm with doubly-exponential running time and a doubly-exponential upper bound on the necessary lookahead~\cite{HoltmannKaiserThomas12}. Recently, the tight exponential bounds on the running time and on the lookahead mentioned above were proven~\cite{KleinZimmermann16a}.
Delay games
with context-free winning conditions turned out to be
undecidable for very small fragments~\cite{FridmanLoedingZimmermann11}. The results of going beyond the $\omega$-regular case by considering WMSO$+$U winning conditions are mentioned above. Furthermore, all delay games with Borel winning conditions
are determined~\cite{KleinZimmermann15}. Finally,
from a more theoretical point of view, Holtman et al.\ also
showed that delay games are a suitable representation of
uniformization problems for relations by continuous functions~\cite{HoltmannKaiserThomas12}.\medskip

\fussy

%%%%%%%%%%%%%%%%%%%%%%%%%%%%%%%%%%%%%%%%%%%%%%%%%%%%%%%%%%%%
%%%%%%%%%%%%%%%%%%%%%%%%%%%%%%%%%%%%%%%%%%%%%%%%%%%%%%%%%%%%
\section{Preliminaries}
\label{sec_prelim}
The set of non-negative (positive) integers is denoted by $ \nats $ ($\natsplus$). An alphabet~$ \Sigma $ is a non-empty finite set of letters, $ \Sigma^{*}
$ is the set of finite words over $ \Sigma $, $ \Sigma^{i} $ the set of
words of length~$ i $, and $ \Sigma^{\omega} $ the set of infinite words. The
empty word is denoted by $ \epsilon $ and the length of a finite word~$ w $
by~$ \size{w} $. For $ w \in \Sigma^{*} \cup \Sigma^{\omega} $ we
write $ w(i) $ for the $ i $-th letter of $ w $.	Given two infinite words $\alpha \in \SigmaI^\omega$ and $\beta  \in \SigmaO^\omega$ we write ${ \alpha \choose \beta}$ for the word ${\alpha(0) \choose \beta(0) } {\alpha(1) \choose \beta(1) }  {\alpha(2) \choose \beta(2) } \cdots  \in (\SigmaI \times \SigmaO)^\omega$. Analogously, we write ${x \choose y}$ for finite words $x$ and $y$, provided they are of equal length.

%%%%%%%%%%%%%%%%%%%%%%%%%%%%%%%
%%%%%%%%%%%%%%%%%%%%%%%%%%%%%%%
\subsection{Parity Games}
\label{subsec_paritygames}
An arena~$\arena$ is a tuple~$(V, V_I, V_O, E)$, where $(V,E)$ is a finite directed graph without terminal vertices and $\{V_I, V_O\}$ is a partition of $V$ into the positions of Player~$I$ and Player~$O$. A parity game~$\game = (\arena, \col)$ consists of an arena~$\arena$ with vertex set~$V$ and of a priority function~$\col \colon V \rightarrow \nats$. A play~$\rho$ is an infinite sequence~$v_0 v_1 v_2 \cdots$ of vertices such that $(v_i, v_{i+1}) \in E$ for all $i$. A strategy for Player~$O$ is a map~$\sigma \colon V^*V_O \rightarrow V$ such that $(v_i, \sigma(v_0 \cdots v_i)) \in E$ for all $v_i \in V_O$. The strategy~$\sigma$ is positional, if $\sigma(wv) = \sigma(v)$ for all $wv \in V^*V_O$. Hence, we denote it as mapping from $V_O$ to V. A play~$v_0 v_1 v_2 \cdots $ is consistent with $\sigma$, if $v_{i+1} = \sigma(v_0 \cdots v_i)$ for every $i$ with $v_i \in V_O$. The strategy~$\sigma$ is winning from a vertex~$v \in V$, if every play~$v_0 v_1 v_2 \cdots $ with $v_0 = v$ that is consistent with $\sigma$ satisfies the parity condition, i.e., the maximal priority appearing infinitely often in $\col(v_0) \col(v_1) \col(v_2) \cdots $ is even. The definition of (winning) strategies for Player~$I$ is dual. Parity games are positionally determined~\cite{EmersonJutla91,Mostowski91}, i.e., from every vertex one of the players has a positional winning strategy.

%%%%%%%%%%%%%%%%%%%%%%%%%%%%%%%
%%%%%%%%%%%%%%%%%%%%%%%%%%%%%%%
\subsection{Delay Games}
\label{subsec_delaygames}
A delay function is a mapping $ f \colon \nats \rightarrow\natsplus$, which is said to be constant, if $ f(i) = 1 $ for
every $ i > 0 $. Given a winning condition~$ L \subseteq
  \left( \SigmaI \times \SigmaO \right)^{\omega} $ and a delay
function~$ f $, the game $ \delaygame{L} $ is played by two players,
 Player~$ I $ and Player~$ O $, in rounds $ i = 0,1,2,\ldots $ as follows: in round~$ i $, Player~$ I $
picks a word $ u_{i} \in \SigmaI^{f(i)} $, then Player~$ O $ picks one
letter $ v_{i} \in \SigmaO $. We refer to the sequence $
(u_{0},v_{0}), (u_{1},v_{1}), (u_{2},v_{2}), \ldots$ as a play of $
\delaygame{L} $. Player~$ O $ wins the play if the outcome~$ {
  u_{0}u_{1}u_{2} \cdots \choose v_{0}v_{1}v_{2} \cdots }$ is in $ L $, otherwise Player~$ I $ wins.

Given a delay function $ f $, a strategy for Player~$ I $ is a mapping
$ \stratI \colon \SigmaO^{*} \rightarrow\SigmaI^{*} $ where $
\size{\stratI(w)} = f(\size{w}) $, and a strategy for Player~$ O $ is
a mapping $ \stratO \colon \SigmaI^{*} \rightarrow\SigmaO $. Consider a
play $ (u_{0},v_{0}), (u_{1},v_{1}), (u_{2},v_{2}), \ldots $ of $
\delaygame{L} $. Such a play is consistent with $ \stratI $, if $
u_{i} = \stratI(v_{0} \cdots v_{i-1}) $ for every $ i \in \nats $. It
is consistent with $ \stratO $, if $ v_{i} = \stratO(u_{0} \cdots
u_{i}) $ for every $ i \in \nats $. A strategy~$ \strat $ for Player~$\p\in \set{I,O}$ is winning, if every play that is consistent with $\strat$ is winning for Player~$\p$. We say that a player wins $ \delaygame{L} $, if she has a winning strategy.

%%%%%%%%%%%%%%%%%%%%%%%%%%%%%%%
%%%%%%%%%%%%%%%%%%%%%%%%%%%%%%%
\subsection{Prompt LTL}
\label{subsec_prompt}
Fix a set $\AP$ of atomic propositions. $ \prompt $ formulas are given by
%Let $\AP$ be a finite set of atomic propositions. The formulas of $\prompt$ are given by the grammar
%
\begin{equation*}\phi \cceq p \mid \neg p \mid \phi \wedge \phi \mid \phi \vee
\phi \mid \X \phi \mid \F \phi \mid \G \phi \mid \phi \U \phi \mid \phi \R \phi \mid \Fp \phi
  ,\end{equation*}
where $p \in \AP$. We use $\phi \rightarrow \psi$ as shorthand for $\neg \phi \vee \psi$,
where we require $\phi$ to be a $\Fp$-free formula (for which the negation can be pushed to the atomic propositions using the dualities of the classical temporal operators). The size~$\size{\phi}$ of $\phi$ is the number of subformulas of $\phi$.

The satisfaction 
relation is defined for an $\omega$-word~$w \in \left( \pow{\AP} \right)^{ \omega }$, a
position~$ i $ of $w$, a bound~$k$ for the prompt eventually operators, and a~$\prompt$ formula. The definition is standard for the classical operators and defined as follows for the prompt eventually:
\begin{eqnarray*}
&(w,i,k)\models\Fp\phi \text{ if, and only if, there exists a } j \\
& \text{ with } 0\le j \le k \text{ such that } (w,i+j,k)\models\phi .
\end{eqnarray*}
For the sake of brevity, we write $(w,k) \models \phi$ instead of
$(w,0,k) \models \phi$. Note that $\phi$ is an $\ltl$ formula~\cite{Pnueli77}, if it does not contain the prompt eventually operator. Then, we write $w \models \phi$.

%%%%%%%%%%%%%%%%%%%%%%%%%%%%%%%
%%%%%%%%%%%%%%%%%%%%%%%%%%%%%%%
\subsection{The Alternating-color Technique}
\label{subsec_altcolor}
Let $p\notin \AP$ be a fixed fresh proposition. An
$\omega$-word~$w'\in\left(2^{\AP\cup\{p\}}\right)^{\omega}$ is a $p$-coloring of
$w\in\left(2^{\AP}\right)^{\omega}$ if $w'(i)\cap \AP=w(i)$ for all $i$.

A position~$i$ of a word in $\left(2^{\AP\cup\set{p}}\right)^{\omega}$ is a change point, if $i=0$ or if the truth value of $p$ at positions~${i-1}$
and $i$ differs. A $p$-block is an infix~$w'(i) \cdots w'({i+j})$ of $w'$ such that $i$ and $i+j+1$ are adjacent change points. Let $k \ge 1$: we say
that $w'$ is $k$-spaced, if $w$ has infinitely many changepoints and each
$p$-block has length at least $k$; we say that $w'$ is $k$-bounded, if each
$p$-block has length at most~$k$ (which implies that $w'$ has infinitely many change points).

Given a $\prompt$ formula~$\phi$,
let $\rel'(\phi)$ denote the formula obtained by inductively replacing
every subformula~$\Fp\psi$ by
\begin{equation*}
(p\rightarrow (p\U(\neg p\U\rel'(\psi))))\wedge(\neg p\rightarrow (\neg p\U(
p\U\rel'(\psi))))
\end{equation*}
and let $\rel(\phi) = \rel'(\phi) \wedge \G\F p \wedge \G\F \neg p$, i.e., we additionally require infinitely many change points.
Intuitively, instead of requiring $\psi$ to be satisfied within a bounded number of steps, $\rel(\phi)$ requires it to be satisfied within at most one change point. The relativization~$\rel(\phi)$ is an $\ltl$ formula of size~$\bigo(\size{\phi})$. Kupferman et al.\ showed that $\phi$ and $\rel(\phi)$ are \myquot{equivalent} on $\omega$-words
which are bounded and spaced.

\begin{lemma}[\cite{KupfermanPitermanVardi09}]
\label{lemma_alternatingcolor}
Let $\phi$ be a $\prompt$ formula and $k \in \nats$.
\begin{enumerate}
\item \label{lemma_alternatingcolor_prompttoltl}
If $(w,k)\models \phi$, then $w' \models \rel(\phi)$ for
every $k$-spaced $p$-coloring~$w'$ of $w$.

\item \label{lemma_alternatingcolor_ltltoprompt}
 If $w'$ is a $k$-bounded $p$-coloring of $w$ such that
$w' \models \rel(\phi)$, then $(w,2k)\models\phi$.
\end{enumerate}
\end{lemma}

%%%%%%%%%%%%%%%%%%%%%%%%%%%%%%%%%%%%%%%%%%%%%%%%%%%%%%%%%%%%
%%%%%%%%%%%%%%%%%%%%%%%%%%%%%%%%%%%%%%%%%%%%%%%%%%%%%%%%%%%%
\section{Delay Games with Prompt-LTL Winning Conditions}
\label{sec_promptdelaygames}
In this section, we study delay games with $\prompt$ winning conditions. Player~$O$'s goal in such games is to satisfy the winning condition~$\phi$ with respect to a bound~$k$ which is uniform among all plays consistent with the strategy. We show that such games are reducible to delay games with $\ltl$ winning conditions by tailoring the alternating-color technique to delay games and integrating it into the algorithm for solving $\omega$-regular delay games~\cite{KleinZimmermann16a}.

Throughout this section, we fix a partition~$\AP = I \cup O$ of the set of atomic propositions into input propositions~$I$ under Player~$I's$ control and output propositions~$O$ under Player~$O$'s control. Let $\SigmaI= \pow{I}$ and $\SigmaO= \pow{O}$ and let ${\alpha \choose \beta} \in (\SigmaI \times \SigmaO)^\omega$ be given. We write $({\alpha \choose \beta},k) \models \phi$ for 
\begin{equation*}
  ((\alpha(0) \cup \beta(0))\,(\alpha(1) \cup \beta(1))\,(\alpha(2) \cup \beta(2))\,\cdots, k) \models \phi.
\end{equation*}
Given $\varphi$ and a bound~$k$, we define 
$L(\phi, k) = \set{ 
{ \alpha \choose \beta }
 \in  (\SigmaI \times \SigmaO )^\omega 
\mid ({\alpha \choose \beta},k) \models \phi }$.
If $\phi$ is an $\ltl$ formula, then this language is independent of $k$ and will be denoted by $L(\phi)$.

A $\prompt$ delay game~$\delaygame{\phi}$ consists of a delay function~$f$ and a $\prompt$ formula~$\phi$. We say that Player~$\p \in \set{I,O}$ wins $\delaygame{\phi}$ for the bound $k$, if she wins $\delaygame{L(\phi, k)}$. If we are not interested in the bound itself, but only in the existence of some bound, then we also say that Player~$O$ wins $\delaygame{\phi}$, if there is some $k$ such that she wins $\delaygame{\phi}$ for $k$.  If $\phi$ is an $\ltl$ formula, then we call $\delaygame{\phi}$ an $\ltl$ delay game. The winning condition~$L(\phi)$ of such a game is $\omega$-regular and independent of $k$.

In this section, we solve the following decision problem: given a $\prompt$ formula~$\phi$, does Player~$O$ win $\delaygame{\varphi}$ for some delay function~$f$? Furthermore, we obtain upper bounds on the necessary lookahead and the necessary bound~$k$, which are complemented by matching lower bounds in the next section.

With all definitions at hand, we state our main theorem of this section. 

\begin{theorem}
\label{thm_promptdelaycomplexity}
The following problem is in $\threeexp$: given a $\prompt$ formula~$\phi$, does Player~$O$ win $\delaygame{\phi}$ for some delay function~$f$? 
\end{theorem}

\begin{proof}	
We reduce $\prompt$ to $\ltl$ delay games using the alternating-color technique. To this end, we add the proposition~$p$, which induces the coloring, to $O$, i.e., in a game with winning condition~$\rel(\varphi)$ Player~$O$'s alphabet is $\pow{O \cup \set{p}}$. In Lemma~\ref{lemma_splitgamecorrectness}, we prove that Player~$O$ wins $\delaygame{\phi}$ for some delay function~$f$ if, and only if, Player~$O$ wins $\delaygame{\rel(\phi)}$ for some delay function~$f$. This equivalence proves our claim: Determining whether Player~$O$ wins a delay game (for some $f$) whose winning condition is given by a deterministic parity automaton is $\exptime$-complete~\cite{KleinZimmermann16a}. We obtain an algorithm with triply-exponential running time by constructing a doubly-exponential deterministic parity automaton recognizing $L(\rel(\phi))$ and then running the exponential-time algorithm on it. \qed
\end{proof}

Thus, it remains to prove the equivalence between the delay games with winning conditions $\phi$ and $\rel(\phi)$. The harder implication is the one from the $\ltl$ delay game to the $\prompt$ delay game. There is a straightforward extension of the solution to the delay-free case. There, one proves that a finite-state strategy for the {\ltl} game with winning condition~$\rel(\phi)$ (which always exists, if Player~$O$ wins the game) only produces $k$-bounded outcomes, for some $k$ that only depends on the size of the strategy. Hence, by projecting away the additional proposition~$p$ inducing the coloring, we obtain a winning strategy for the {\prompt} game with winning condition~$\phi$ with bound~$2k$ by applying Lemma~\ref{lemma_alternatingcolor}.\ref{lemma_alternatingcolor_ltltoprompt}.

\sloppy

Now, consider the case with lookahead: if Player~$O$ wins $\delaygame{\rel(\phi)}$, which has an $\omega$-regular winning condition, then also $\delaygamep{\rel(\phi)}$ for some triply-exponential constant $f'$~\cite{KleinZimmermann16a}. We can model $\delaygamep{\rel(\phi)}$ as a delay-free parity game of \mbox{quadruply-exponential} size by storing the lookahead explicitly in the state space of the parity game. A positional winning strategy in this parity game only produces $k$-bounded plays, where $k$ is the size of the delay-free game, as the color has to change infinitely often. Hence, such a strategy can be turned into a winning strategy for Player~$O$ in $\delaygamep{\phi}$ with respect to some quadruply-exponential bound~$k$. However, this naive approach is not optimal: we present a more involved construction that achieves a triply-exponential bound~$k$. The problem with the aforementioned approach is that the decision to produce a change point depends on the complete lookahead. We show how to base this decision on an exponentially smaller abstraction of the lookaheads, which yields an asymptotically optimal bound~$k$. To this end, we extend the construction underlying the algorithm for $\omega$-regular delay games~\cite{KleinZimmermann16a} by integrating the alternating-color technique.

\fussy

Intuitively, we assign to each $w \in \SigmaI^*$, i.e., to each potential additional information Player~$O$ has access to due to the lookahead, the behavior $w$ induces in a deterministic automaton~$\aut$ accepting $L(\rel(\phi))$, namely the state changes induced by $w$ and the most important color on these runs. We construct $\aut$ such that it keeps track of change points in its state space, which implies that they are part of the behavior of $w$. Then, we construct a parity game in which Player~$I$ picks such behaviors instead of concrete words over $\SigmaI$ and Player~$O$ constructs a run on suitable representatives. The resulting game is of triply-exponential size and a positional winning strategy for this game can be turned back into a winning strategy for $\delaygame{\phi}$ satisfying asymptotically optimal bounds on the initial lookahead and the bound~$k$. Thus, we save one exponent by not explicitly considering the lookahead, but only its effects.

We first extend the construction of a delay-free parity game~$\game$ that has the same winner as $\delaygame{\rel(\phi)}$ from~\cite{KleinZimmermann16a}. The extension is necessary to obtain a \myquot{small} bound $k$ when applying the alternating-color technique, which turns a positional winning strategy for $\game$ into a winning strategy for Player~$O$ in $\delaygamep{\phi}$ 
for some $f'$. 

To this end, let $\aut = (Q, \SigmaI \times \SigmaO, q_I, \delta, \col)$ be a deterministic max-parity automaton\footnote{Recall that $\col \colon Q \rightarrow \nats$ is a coloring of the states and a run $q_0 q_1 q_2 \cdots$ is accepting, if the maximal color occurring infinitely often in $\col(q_0) \col(q_1) \col(q_2) \cdots$ is even. See, e.g., \cite{GraedelThomasWilke02} for details.} recognizing $L(\rel(\phi))$. First, as in the original construction, we add a deterministic monitoring automaton to keep track of certain information of runs. In the $\omega$-regular case~\cite{KleinZimmermann16a}, this information is the maximal priority encountered during a run. Here, we additionally need to remember whether the input word contains a change point. Let $T_0 = \pow{\set{p}}$ and $T = T_0 \times \set{0,1}$. Furthermore, for $(t,s) \in T$ and $t' \in T_0$, we define the update~$\update((t,s),  t') \in T$ of $(t,s)$ by $t'$ to be $(t', s')$, where $s' =0$ if, and only if, $s=0$ and $t = t'$. Intuitively, the first component of a tuple in $T$ stores the last truth value of $p$ and the second component is equal to one if, and only if, there was a change point.

Now, we define the deterministic parity automaton
\begin{equation*} 
  \autc = (Q_\autc, \SigmaI \times \SigmaO, q_I^\autc, \delta_\autc, \col_\autc)
\end{equation*} 
with $Q_\autc = Q \times \col(Q) \times T$, $q_I^\autc = (q_I, \col(q_I), (t',0))$ for some arbitrary $t' \in T_0$, $\col(q,m,t) = m$, and $\delta_\autc((q, m, t),{a \choose b}) = (q', \max\set{m, \col(q')}, \update(t,  b \cap \set{p}))$ with $q ' = \delta(q, {a \choose b})$.

First, let us note that $\autc$ does indeed keep track of the information described above. 

\begin{remark}
\label{rem_trackingautomaton}
Let $w \in (\SigmaI \times \SigmaO)^+$ and let $(q_0,m_0,t_0)\cdots(q_{\size{w}},m_{\size{w}},t_{\size{w}})$ be the run of $\autc$ on $w$ starting in $(q_0,m_0,t_0)$ such that $m_0 = \col(q_0)$ and $t_0 = (t_0',0)$ for some $t_0' \in T_0$. Then, $q_0 q_1 \cdots q_{\size{w}}$ is the run of $\aut$ on $w$ starting in $q_0$, $m_{\size{w}} = \max\set{\col(q_j) \mid 0 \le j \le \size{w}}$, and $t_{\size{w}} = (t_{\size{w}}',s_{\size{w}})$ such that $t_{\size{w}}'$ is the color of the last letter of $w$, and such that $s_{\size{w}} = 0$ if, and only if, all letters of $w$ have color~$t_0'$. In particular, if $w$ is preceded by a word whose last letter has color~$t_0'$, then there is a change point in $w$ if, and only if, $s_{\size{w}} =1$.  
\end{remark}

Next, we classify possible moves~$w \in \SigmaI^*$ according to the behavior they induce on~$\autc$. Let $\delta_\autp \colon \pow{Q_\autc} \times \SigmaI \rightarrow \pow{Q_\autc}$ denote the transition function of the power set automaton of the projection of $\autc$ to $\SigmaI$, i.e., $\delta_\autp(S,a) = \set{ \delta_\autc(q,{a \choose b} ) \mid q \in S \text{ and } b \in \SigmaO }$. As usual, we define $\delta^*_\autp \colon \pow{Q_\autc} \times \SigmaI^* \rightarrow \pow{Q_\autc}$ inductively via $\delta^*_\autp(S,\epsilon) = S$ and $\delta^*_\autp(S,wa) = \delta_\autp(\delta^*_\autp(S,w),a) $. 

Let $D \subseteq Q_\autc$ be a non-empty set and let $w \in \SigmaI^+$. We define the function~$\resolve_w^D \colon D \rightarrow \pow{Q_\autc}$ via
\[
\resolve_w^D(q,m,(t,s)) = \delta^*_\autp(\set{\,(q, \col(q), (t,0))\,}, w)
\]
for every $(q,m,(t,s)) \in D$. Note that we use $\col(q)$ and $(t,0)$ as the second and third component in the input for $\delta_\autp^*$, not $m$ and $(t,s)$ from the input to $\resolve_w^D$. This resets the tracking components of $\autc$. If we have $(q',m',(t',s')) \in \resolve_w^D(q,m,(t,s))$, then there is a word~$w'$ over $\SigmaI  \times \SigmaO$ whose projection to $\SigmaI$ is $w$ and such that the run of $\aut$ processing $w'$ from $q$ has the maximal priority~$m'$, $t'$ is the color of the last letter of $w$, and  $s'$ encodes the existence of change points in $w'$, as explained in Remark~\ref{rem_trackingautomaton}. Thus, this function captures the behavior induced by $w$ on $\autc$. We allow to restrict the domain of such a function, as we do not have to consider every possible state, only those that are reachable by the play prefix constructed thus far. 

Let $\resolve \colon Q_\autc \rightarrow \pow{Q_\autc}$ be a partial function. We say that $w$ is a witness for $r$, if $\resolve_w^{\dom(r)} =\resolve $. Thus, we can assign a language $W_\resolve \subseteq \SigmaI^*$ of witnesses to each such $\resolve$. Let $\curlyR$ denote the set of such functions~$\resolve$ with infinite witness language~$W_\resolve$. If $w$ is a witness of $\resolve \in \curlyR$, then $r$ encodes the state transformations induced by $w$ in the projection of $\aut$ to $\SigmaI$ as well as the maximal color occurring on these runs and the existence of change points on these. The latter is determined by the letters projected away, but still stored explicitly in the state space of the automaton. Furthermore, as we require $r \in \curlyR$ to have infinitely many witnesses, there are arbitrarily long words with the same behavior. On the other hand, the language~$W_\resolve$ of witnesses of $\resolve$ is recognizable by a DFA of size $2^{n^2}$~\cite{KleinZimmermann16a}, where $n$ is the size of $\autc$. Hence, every $\resolve$ also has a witness of length at most $2^{n^2}$. This allows to replace \emph{long} words~$w \in \SigmaI^{*}$ by equivalent ones that are bounded exponentially in $n$. 

\sloppy

Next, we define a delay-free parity game in which Player~$I$ picks functions~\mbox{$\resolve_i \in \curlyR$}  while Player~$O$ picks states~$q_i$ such that there is a word~$w'_i$ in $(\SigmaI \times \SigmaO)^*$ whose projection to $\SigmaI$ is a witness of $\resolve_i$ and such that $w'_{i}$ leads $\autc$ from $q_i$ to $q_{i+1}$. By construction, this property is independent of the choice of the witness. Furthermore, to account for the delay, Player~$I$ is always two moves ahead. Thus, instead of picking explicit words over their respective alphabets, the players pick abstractions, Player~$I$ explicitly and Player~$O$ implicitly by constructing the run.

\fussy

Formally we define the parity game~$\game = ((V, V_O, V_I, E), \col')$ where $V = V_I \cup V_O$, $V_I = \set{v_I} \cup \curlyR \times Q_\autc$ with the designated initial vertex~$v_I$ of the game, and
$V_O = \curlyR$. Further, $E$ is the union of the following sets of edges: initial moves
$\set{ (v_I, \resolve) \mid \dom(\resolve) = \set{q_I^\autc}}$ for Player~$I$,
regular moves $\set{ ((\resolve,q), \resolve') \mid \dom(\resolve') = \resolve(q) }$ for Player~$I$, and moves
 $\set{ (\resolve, (\resolve, q)) \mid q \in \dom(\resolve)}$ for Player~$O$.
Finally, $\col'(v) = m$, if $v = (\resolve,(q,m,s)) \in \curlyR \times Q_\autc$, and zero otherwise.

This finishes the construction of the game~$\game$. The following lemma states the relation between $\game$ and the delay games with winning conditions $\phi$ and $\rel(\phi)$ and implies the equivalence of the delay games with winning conditions $\phi$ and $\rel(\phi)$.
 
\begin{lemma}
\label{lemma_splitgamecorrectness}	
Let $n = \size{Q_\autc}$, where $Q_\autc$ is the set of states of $\autc$ as defined above.
\begin{enumerate} 
	\item\label{lemma_splitgamecorrectness_prompt2ltl}	
If Player~$O$ wins $\delaygame{\phi}$ for some delay function~$f$, then also $\delaygame{\rel(\phi)} $ for the same $f$.

%	\item\label{lemma_splitgamecorrectness_ltl2prompt}	
%If Player~$O$ wins $\delaygame{\rel(\phi)} $ for some delay function~$f$, then also $\delaygame{\phi}$ for the same $f$.
	
	\item\label{lemma_splitgamecorrectness_ltl2parity}	
If Player~$O$ wins $\delaygame{\rel(\phi)} $ for some delay function~$f$, then also $\game$.
	
	\item\label{lemma_splitgamecorrectness_parity2ltl}	
If Player~$O$ wins $\game$, then also $\delaygame{\phi} $ for the constant delay function~$f$ with $f(0) = 2^{n^2+1}$ and some bound~$k \le 2^{2n^2+2}$.

\end{enumerate}
\end{lemma}

First consider that the following lemma proved for the $\omega$-regular case holds in our setting as well. 

\begin{lemma}[\cite{KleinZimmermann16a}]
\label{lemma_resolveprops}
Let $\curlyR$ and $n$ be defined as in Section~\ref{sec_promptdelaygames}.
\begin{enumerate}
	\item\label{lemma_resolveprops_disjointness}
	Let $r \neq r' \in \curlyR$ with $\dom(r) = \dom(r')$. Then, $W_\resolve \cap W_{\resolve'} = \emptyset$.
	 \item\label{lemma_resolveprops_lowerbound}
	 Let $D \subseteq Q_\autc$ be non-empty and let $w \in \SigmaI^*$ with $\size{w} \ge 2^{n^2}$. Then, there is some $\resolve \in \curlyR$ with $\dom(\resolve) = D$ and $w \in W_\resolve$.
\end{enumerate}	
\end{lemma}

The first statement follows from $r^D_w$ being uniquely determined by $D$ and $w$ (as $\autc$ is fixed). Furthermore, the witness language of a partial function~$\resolve \colon Q_\autc \rightarrow \pow{Q_\autc}$ can be recognized by a DFA obtained from composing $\dom(r) \le n$ copies of the power set automaton of the projection of $\autc$ to $\SigmaI$, which has size~$2^n$. Hence, its language is infinite if, and only if, it accepts a word of length at least~$2^{n^2}$. This explains the bound in the second statement of the previous lemma, which also induces the upper bound in the third implication to be proven.

Fix some non-empty $D$. Due to Lemma~\ref{lemma_resolveprops}, we can construct a function~$\resolve_D$ that maps long enough words $w$ to the unique function~$\resolve \in \curlyR$ with $\dom(\resolve) = D$ and $w \in W_\resolve$.

Now, we are ready to prove Lemma~\ref{lemma_splitgamecorrectness}:

\begin{proof} \ref{lemma_splitgamecorrectness_prompt2ltl}.) 
Let $\stratO \colon \SigmaI^* \rightarrow \SigmaO$ be a winning strategy for Player~$O$ in $\delaygame{\varphi}$ with bound~$k$. We define the strategy~$\stratO'$ for $\delaygame{\rel(\phi)}$ via \[
\stratO'(w) =	
\begin{cases}
\stratO(w)&\text{if $\size{w} \bmod 2k < k$,}\\
\stratO(w) \cup \set{p}&\text{if $\size{w} \bmod 2k \ge k$,}	
\end{cases}
\]
 i.e., the strategy mimics the moves of $\stratO$ and additionally produces $k$-bounded and $k$-spaced $p$-blocks. In particular, every outcome of $\stratO'$ is a $k$-spaced coloring of an outcome of $\stratO$. As every such outcome of $\stratO$ satisfies $\phi$ with respect to $k$, applying Lemma~\ref{lemma_alternatingcolor}.\ref{lemma_alternatingcolor_prompttoltl} yields that $\stratO'$ is  winning  for Player~$O$ in $\delaygame{\rel(\phi)}$.

%%%%%%%%%%%%%%%%%%%%%%%%%%%%%%%%%%%%%
%%%%%%%%%%%%%%%%%%%%%%%%%%%%%%%%%%%%%
\ref{lemma_splitgamecorrectness_ltl2parity}.) This result is proven in~\cite{KleinZimmermann16a} for a tracking automaton~$\autc$ without the last component tracking the change points. However, this additional component is inconsequential for this implication, as the tracking is passive and does not influence the winner of $\game$: a winning strategy for Player~$O$ in $\delaygame{\rel(\phi)}$ can be simulated in $\game$ in order to obtain a winning strategy for Player~$O$ in this game. 

%%%%%%%%%%%%%%%%%%%%%%%%%%%%%%%%%%%%%
%%%%%%%%%%%%%%%%%%%%%%%%%%%%%%%%%%%%%
\ref{lemma_splitgamecorrectness_parity2ltl}.)
Let Player~$O$ win $\game$. Due to postional determinacy of parity games~\cite{EmersonJutla91,Mostowski91}, Player~$O$ has a positional winning strategy~$\sigma \colon V_O \rightarrow V$ from $v_I$ for $\game$. First, we turn $\sigma$ into a winning strategy~$\stratO'$ for Player~$O$ in $\delaygame{\rel(\phi)}$ for some constant delay function~$f$ and show that $\stratO'$ only produces $k$-bounded plays for some suitable $k$. This strategy is then turned into a winning strategy for Player~$O$ in $\delaygame{\phi}$.

First, fix the constant delay function~$f$ with $f(0) = 2d$, where $d = 2^{n^2}$. Note that $d$ is the lower bound on words~$w \in \SigmaI$ witnessing an $\resolve \in \curlyR$ (Lemma~\ref{lemma_resolveprops}).

We define the winning strategy $\stratO'$ for Player~$O$ in $\delaygame{\rel(\phi)}$ inductively by simulating a play in $\delaygame{\rel(\phi)}$ by a play in $\game$. In the following, the players will pick their moves in $\delaygame{\rel(\phi)}$ in blocks of length~$d$. We will denote Player~$I$'s blocks by $\block{a}$ and Player~$O$'s blocks by $\block{b}$. 

In round~$0$ of $\delaygame{\rel(\phi)}$, Player~$I$ picks $\block{a_0} \block{a_1} $. We define $q_0 = q_I^\autc$, $\resolve_0 = \resolve_{\set{q_0}}(\block{a_0})$ and $\resolve_1 = \resolve_{\resolve_0(q_0)}(\block{a_1})$. Then, $ v_I\, \resolve_0\, (\resolve_0, q_0)\, \resolve_1$ is a play prefix of $\game$ that is consistent with $\sigma$.  

Thus, we are in the following situation for $i=1$: in $\delaygame{\rel(\phi)}$, Player~$I$ has picked $\block{a_0} \cdots \block{a_i}$ and Player~$O$ has picked $\block{b_0} \cdots \block{b_{i-2}}$, and in $\game$, we have constructed a play prefix~$v_I\, \resolve_0\, (\resolve_0, q_0)\, \resolve_1\, \cdots\, (\resolve_{i-1}, q_{i-1})\, r_i$ consistent with $\sigma$. Also, every $\block{a_j}$ is a witness for $\resolve_j$.

In this situation for an arbitrary $i$, let $q_i$ be the state of $\autc$ such that $\sigma(\resolve_i) = (r_i, q_i)$.  We have $q_i \in \dom(\resolve_i) = \resolve_{i-1}(q_{i-1})$. Due to $\block{a_{i-1}}$ being a witness of $\resolve_{i-1}$, 
there is at least one $\block{b_{i-1}}$ such that $\autc$ reaches $q_i$ when processing ${\block{a_{i-1}} \choose \block{b_{i-1}}}$ from $(q_{i-1}', \col(q_{i-1}'), (t_{i-1},0))$, where $q_{i-1} = (q_{i-1}', m_{i-1}, (t_{i-1}, s_{i-1}) )$. We define $\stratO'$ to pick $\block{b_{i-1}}$ during the next $d$ rounds. While doing this, Player~$I$ picks the next block~$\block{a_{i+1}}$. We define $\resolve_{i+1} = \resolve_{\resolve_i(q_i)}(\block{a_{i+1}})$. Then, we are again in the situation as described above for $i+1$. 

Let $w = 
{ \block{a_0} \choose \block{b_0} }
{ \block{a_1} \choose \block{b_1} }
{ \block{a_2} \choose \block{b_2} }\cdots$ be an outcome of a play that is consistent with $\stratO'$ and let $ \rho = v_I\, \resolve_0\, (\resolve_0, q_0)\, \resolve_1\,  (\resolve_{1}, q_{1})\, r_2 \,\cdots$ be the play in $\game$ constructed during the simulation, which is consistent with $\sigma$. Finally, let $q_i=(q_i',m_i,(t_i,s_i))$ for every $i$. 

A straightforward induction using Remark~\ref{rem_trackingautomaton} shows that $q_{i+1}'$ is the state $\aut$ reaches when processing ${\block{a_i} \choose \block{b_i}}$ starting in $q_i'$. Furthermore, $m_{i+1}$ is the largest priority of this run and $s_i$ encodes the existence of a change point in this infix of $w$. As $\rho$ is winning for Player~$O$, $m_0 m_1 m_2,  \cdots$ satisfies the parity condition. Thus, by the characterization of $m_i$ above, $\aut$ accepts $w$. In particular, every outcome of a play that is consistent with $\stratO'$ satisfies $\rel(\phi)$.

To be able to apply the alternating-color technique, it remains to show that every such outcome is $k$-bounded by some uniform $k$. Thus, let $w$ and $\rho$ be as above. First, we show that if $s_i = s_{i+1} = \cdots s_{i + i'} = 0$ for some $i>0$, then $i' < \size{\curlyR} = \size{V_O}$. Assume the opposite. Then, there are~$j,j'$ with $i \le j < j' \le i'$ such that $\resolve_j = \resolve_{j'}$. This implies~$q_j = q_{j'}$, as these states are uniquely determined by applying $\sigma$ to $\resolve_j = \resolve_{j'}$. Thus, the play
\[
\rho' = v_I\, \resolve_0\, (\resolve_0, q_0)\, \resolve_1 \,\cdots\, (\resolve_{j-1}, q_{j-1})\, \Big( \, \resolve_j \,(\resolve_j, q_j)\, \cdots\, \resolve_{j'-1}\,(\resolve_{j'-1}, q_{j'-1}) \,\Big)^\omega 
\] 
is consistent with $\sigma$ as well, as the strategy is positional. 

Consider the word \[
w' = 
{ \block{a_0} \choose \block{b_0} }
{ \block{a_1} \choose \block{b_1} }
\cdots
{ \block{a_{j-1}} \choose \block{b_{j-1}} }\,\left[
{ \block{a_j} \choose \block{b_j} }
\cdots 
{ \block{a_{j'-1}} \choose \block{b_{j'-1}} }
 \right]^\omega.\]

Using a similar reasoning as above, one can show that $w$ is accepted by $\aut$, as the accepting run is encoded in the winning play~$\rho'$.  Furthermore, by construction of $\rho'$, all $s_i$ contained in the states of the loop of $\rho'$ are equal to zero. Hence, there are only finitely many change points in $w'$, as the $s_i$ keep track of the change points. This yields the desired contradiction to the fact that every word in $L(\aut)$ has infinitely change points, as required by $\rel(\phi)$. 

As a consequence, every $(\size{\curlyR}+1)$-th block~${\block{a_i} \choose \block{b_i}}$ of an outcome of $\stratO'$ contains a change point. Since every block has length~$d$, every outcome is $(\size{\curlyR}+1)\cdot d$-bounded. 

Finally, we turn $\stratO' \colon \SigmaI^* \rightarrow \pow{O \cup \set{p}}$ for $\delaygame{\rel(\phi)}$, where $p$ is the proposition introduced in the alternating-color technique, into a strategy~$\stratO \colon \SigmaI^* \rightarrow \pow{O}$ for $\delaygame{\phi}$ by defining $\stratO(w) = \stratO'(w) \cap O$. Every outcome of a play that is consistent with $\stratO$ has a $p$-coloring that is the outcome of a play that is consistent with $\stratO'$, and therefore satisfies $\rel(\phi)$ and is $(\size{\curlyR}+1)\cdot d$-bounded. Thus, Lemma~\ref{lemma_alternatingcolor}.\ref{lemma_alternatingcolor_ltltoprompt} shows that $\stratO$ is winning for Player~$O$ in $\delaygame{\phi}$ with bound
$k = 2(\size{\curlyR}+1)\cdot d \le 2^{2n^2+2}$. \qed
\end{proof}

The automaton~$\aut$ recognizing $L(\rel(\phi))$ can be constructed such that $\size{\aut} \in 2^{2^{\bigo(\size{\phi})}}$, which implies $n \in 2^{2^{\bigo(\size{\phi})}}$, using a standard construction for translating $\ltl$ into non-deterministic Büchi automata and then Schewe's determinization construction~\cite{Schewe09}. Applying all implications of Lemma~\ref{lemma_splitgamecorrectness} yields upper bounds on the neccessary constant lookahead and on the neccessary bound~$k$ on the scope of the prompt eventually operator.

\begin{corollary}
\label{cor_promptupperbounds}
If Player~$O$ wins $\delaygame{\phi}$ for some delay function~$f$ and some $k$, then also for some constant~delay function~$f$ with $f(0) \in 2^{2^{2^{\bigo(\size{\phi})}}}$\!\! and some $k \in {2^{2^{2^{\bigo(\size{\phi})}}}}$\!\! simultaneously.
\end{corollary}

%%%%%%%%%%%%%%%%%%%%%%%%%%%%%%%%%%%%%%%%%%%%%%%%%%%%%%%%%%%%
%%%%%%%%%%%%%%%%%%%%%%%%%%%%%%%%%%%%%%%%%%%%%%%%%%%%%%%%%%%%
\section{Lower Bounds for LTL and Prompt-LTL Delay Games}
\label{sec_lowerbounds}
We complement the upper bounds on the complexity of solving $\prompt$ delay games, on the necessary lookahead, and on the necessary bound~$k$ by proving tight lower bounds in all three cases. The former two bounds already hold for $\ltl$.

All proofs share some similarities which we discuss first. In particular, they all rely on standard encodings of doubly-exponentially large numbers using \emph{small} $\ltl$ formulas and the interaction between the players. Assume $\AP$ contains the propositions~$b_0, \ldots, b_{n-1}, b_I,b_O$ and let $w \in (\pow{\AP})^\omega$ and $i \in \nats$. We interpret $w(i) \cap \set{b_0, \ldots, b_{n-1}}$ as binary encoding of a number in $[0, 2^n-1]$, which we refer to as the address of position $i$. There is a formula~$\psiinc$ of quadratic size in $n$ such that $(w,i) \models \psiinc$ if, and only if, $m +1 \bmod 2^n = m'$, where $m$ is the address of position~$i$ and $m'$ is the address of position~$i+1$. Now, let $\psi_0 = \bigwedge_{j=0}^{n-1} \neg b_j \wedge \G\psiinc$. If $w \models \psi_0$, then the $b_j$ form a cyclic addressing of the positions starting at zero, i.e., the address of position~$i$ is $i \bmod 2^n$.
If this is the case, we define a block of $w$ to be an infix that starts at a position with address zero and ends at the next position with address~$2^n-1$. We interpret the $2^n$ bits~$b_I$ of a block as a number~$x$ in $R = [0, 2^{2^n}-1]$. Similarly, we interpret the $2^n$ bits~$b_O$ of a block as a number~$y$ from the same range~$R$. Furthermore, there are \emph{small} formulas that are satisfied at the start of the $i$-th block if, and only if, $x_i = y_i$ ($x_i < y_i$, respectively). However, we cannot  compare numbers from different blocks for equality with \emph{small} formulas. Nevertheless, if $x_i$ is unequal to $x_{i'}$, then there is a single bit that witnesses this, i.e., the bit is one in $x_i$ if, and only if, it is zero in $x_{i'}$. We will check this by letting one of the players specify the address of such a witness (but not the witness itself). The correctness of this claim is then verifiable by a \emph{small} formula.

%%%%%%%%%%%%%%%%%%%%%%%%%%%%%%%
%%%%%%%%%%%%%%%%%%%%%%%%%%%%%%%
\subsection{Lower Bounds on Lookahead} 
\label{subsec_lowerbounds_la}
Our first result concerns a triply-exponential lower bound on the necessary lookahead in $\ltl$ delay games, which matches the upper bound proven in the previous section. The exponential lower bound~$2^n$ on the necessary lookahead for $\omega$-regular delay games is witnessed by winning conditions over the alphabet~${1, \ldots, n}$. These conditions require to remember letters and to compare them for equality and order~\cite{KleinZimmermann16a}. Here, we show how to adapt the winning condition to the alphabet~$R$, which yields a triply-exponential lower bound~$2^{\size{R}}$. The main difficulty of the proof is the inability of small $\ltl$ formulas to compare letters from~$R$. To overcome this, we exploit the interaction between the players of the game.
  
\begin{theorem}
\label{thm_ltllookahead_lb}
For every $n>0$, there is an $\ltl$ formula $\phi_n$ of size~$\bigo(n^2)$ such that
\begin{itemize}
	
	\item Player~$O$ wins $\delaygame{\phi_n}$ for some delay function~$f$, but
	
	\item Player~$I$ wins $\delaygame{\phi_n}$ for every delay function~$f$ with $f(0) \le 2^{2^{2^{n}}} $.

\end{itemize}
\end{theorem}

\begin{proof}
Fix some $n>0$. In the following, we measure all formula sizes in $n$. Furthermore, let $I = \set{b_0, \ldots, b_{n-1}, b_I, \sharpsym}$ and $O = \set{b_O, \markl, \markr}$. Assume ${\alpha \choose \beta} \in (\SigmaI \times \SigmaO)^\omega$ satisfies $\psi_0$ from above. Then, $\alpha$ induces a sequence~$x_0 x_1 x_2 \cdots \in R^\omega$ of numbers encoded by the bits~$b_I$ in each block. Similarly, $\beta$ induces a sequence~$y_0 y_1 y_2 \cdots  \in R^\omega$.

The winning condition is intuitively described as follows: $x_i$ and $x_{i'}$ with $i < i'$ constitute a bad $j$-pair, if $x_i = x_{i'} = j$ and $x_{i''} < j$ for all $i < i'' <i'$. Every sequence~$x_0 x_1 x_2 \cdots$ contains a bad $j$-pair, e.g., pick $j$ to be the maximal number occurring infinitely often. In order to win, Player~$O$ has to pick $y_0$ such that $x_0 x_1 x_2 \cdots$ contains a bad $y_0$-pair. It is known that this winning condition requires lookahead of length~$2^m$ for Player~$O$ to win, where $m$ is the largest number that can be picked~\cite{KleinZimmermann16a}.

To specify this condition with a small $\ltl$ formula, we have to require Player~$O$ to copy $y_0$ ad infinitum, i.e., to pick $y_i = y_0$ for all $i$, and to mark the two positions constituting the bad~$y_0$-pair. Furthermore, the winning condition allows Player~$I$ to mark one copy error introduced by Player~$O$ by specifying its address by a $\sharpsym$ (which may appear anywhere in $\alpha$). This forces Player~$O$ to implement the copying correctly and thus allows a small formula to check that Player~$O$ indeed marks a bad $y_0$-pair.
Consider the following properties:
\begin{enumerate}
	\item $\sharpsym$ holds at most once. Player~$I$ uses $ \sharpsym $ to specify the address where he claims an error. 

	\item\label{item_marklok} $\markl$ holds at exactly one position, which has to be the start of a block. Furthermore, we require the two numbers encoded by the propositions~$b_I$ and $b_O$ within this block to be equal. Player~$O$ uses $ \markl $ to denote the first component of a claimed bad $j$-pair.
		
	\item\label{item_markrok} $\markr$ holds at exactly one position, which has to be the start of a block and has to appear at a later position than $\markl$. Again, we require the two numbers encoded by this block to be equal. Player~$O$ uses $ \markr $ to denote the second component of the claimed bad $j$-pair.
		
	\item\label{item_spaceok} For every block between the two marked blocks, we require the number encoded by the $b_I$ to be strictly smaller than the number encoded by the $b_O$.
	
	\item If there is a position~$i_\ssharpsym$ marked by $\sharpsym$, then there are no two different positions $i \neq i'$ such that the following two conditions are satisfied: the addresses of $i$, $i'$, and $i_\ssharpsym$  are equal and $b_O$ holds at $i$ if, and only if, $b_O$ does not hold at $i'$. Such positions witness an error in the copying process by Player~$O$, which manifests itself in a single bit, whose address is marked by Player~$I$ at any time in the future.
\end{enumerate}
Each of these properties~$i \in \set{1,2,3,4,5}$ can be specified by an $\ltl$ formula~$\psi_i$ of at most quadratic size. Now, let $\phi_n = ( \psi_0 \wedge \psi_1 ) \rightarrow (\psi_2 \wedge \psi_3 \wedge \psi_4 \wedge \psi_5 )$. We show that Player~$ O $ wins $ \delaygame{\phi_{n}} $ for some triply-exponential constant delay function, but not for any smaller one. 

Fix $n' = {2^{2^n}}$. We begin by showing that Player~$ O $ wins $ \delaygame{\phi_{n}} $ for the constant delay function with
  $ f(0) = 2^n \cdot \pow{n'} $. A simple induction shows
  that every word $ w \in R^*$ of length $ 2^{n'} $ contains a bad $j$-pair for some $j \in R$. Thus, a move $\SigmaI^{f(0)}$ made by Player~$I$ in round~$0$ interpreted as sequence~$x_0 x_1 \cdots x_{2^{n'}-1} \in R^*$ contains a bad $j$-pair for some fixed $j$.  Hence, Player~$O$'s strategy~$\stratO$ produces the sequence~$j^\omega$ and additionally marks the  corresponding bad $j$-pair with $\markl$ and $\markr$. Every outcome of a play that is consistent with $\stratO$ and satisfies $\psi_0$ also satisfies $\psi_2 \wedge \psi_3 \wedge \psi_4 \wedge \psi_5 $, as Player~$O$ correctly marks a bad $j$-pair and never introduces a copy-error. Hence, $\stratO$ is a winning strategy for Player~$O$.

It remains to show that Player~$ I $ wins
  $ \delaygame{\phi_n} $, if
  $ f(0) \leq 2^n \cdot (2^{n'} - 2) \ge 2^{n'} = 2^{2^{2^n}}$. Let $ w_{n'} \in R^*$ be recursively defined via $ w_0 = 0 $ and
  $ w_{j} = w_{j-1}\, j\, w_{j-1} $.  A simple
  induction shows that $w_{n'}$ does not contain a bad $j$-pair, for every $j \in R$, and that
  $ \size{w_{n'}} =  2^{n'} - 1 $.  
  
  Consider the following strategy~$\stratI$ for Player~$I$ in $\delaygame{\phi_n}$: $\strat$ ensures that $\psi_0$ is satisfied by the $b_j$, which fixes them uniquely to implement a cyclic addressing starting at zero. Furthermore, he picks the $b_I$'s so that the sequence of numbers~$x_0 x_1 \cdots x_\ell$ he generates during the first~$2^n$ rounds is a prefix of $w_{n'}$. This is possible, as each $x_i$ is encoded by $2^n$ bits  and by the choice of $f(0)$. As a
  response during the first~$2^n$ rounds, Player~$ O $ determines some number~$ y \in R $. During the next rounds, Player~$ I $ finishes $w_{n'}$ and then picks some fixed~$x \neq y$ ad infinitum (while still implementing the cyclic addressing). In case Player~$O$ picks both markings~$\markl$ and $\markr$ in way that is consistent with properties~\ref{item_marklok}, \ref{item_markrok}, and \ref{item_spaceok} as above, let $y_0 y_1 \cdots, y_{i}$ be the sequence of numbers picked by her up to and including the number marked by $\markr$. If they are not all equal, then there is an address that witnesses the difference between two of these numbers. Player~$I$ then marks exactly one position with the same address using $\sharpsym$. If this is not the case, he never marks a position with $\sharpsym$.
  
Consider an outcome of a play that is consistent with $\stratI$ and let $x_0 x_1 x_2 \cdots \in R^\omega $ and $y_0 y_1 y_2 \cdots\in R^\omega$ be the sequences of numbers induced by the  outcome. By definition of $\stratI$, the antecedent~$\psi_0 \wedge \psi_1$ of $\phi_n$ is satisfied and $x_0 x_1 x_2 \cdots = w_{n'} \cdot x^\omega$ for some $x \neq y_0$. 

If Player~$ O $ never uses her markers $\markl$ and $\markr$ in a way that satisfies $\psi_2 \wedge \psi_3 \wedge \psi_4$, then Player~$ I $ wins the play, as it satisfies the antecedent of $\phi_n$, but not the consequent. Thus, it remains to consider the case where the outcome satisfies $\psi_2 \wedge \psi_3 \wedge \psi_4$. Let $y_0 y_1 \cdots y_{i}$ be the sequence of numbers picked by her up to and including the number marked by $\markr$. Assume we have $y_0 = y_1 = \cdots = y_{i}$. Then, $\markl$ and $\markr$ specify a bad $y_0$-pair, as implied by $\psi_2 \wedge \psi_3 \wedge \psi_4$ and the equality of the $y_j$. As $w_{n'}$ does not contain a bad $y_0$-pair, we conclude $y_0 = x$. However, $\stratI$ ensures $y_0 \neq x$. Hence, our assumption is false, i.e., the $y_j$ are not all equal. In this situation, $\stratI$ marks a position whose address witnesses this difference. This implies that $\psi_5$ is not satisfied, i.e., the play is winning for Player~$I$. Hence, $\stratI$ is winning for him. \qed
\end{proof}

%%%%%%%%%%%%%%%%%%%%%%%%%%%%%%%
%%%%%%%%%%%%%%%%%%%%%%%%%%%%%%%
\subsection{Lower Bounds on the Bound \boldmath$k$} 
\label{subsec_lowerbounds_k}
Our next result is a lower bound on  the necessary bound~$k$ in a $\prompt$ delay game, which is proven by a small adaption of the game constructed in the previous proof. The winning condition additionally requires Player~$O$ to use the mark~$\markr$ at least once and $k$ measures the number of rounds before Player~$O$ does so. It turns out Player~$I$ can enforce a triply-exponential $k$, which again matches the upper bound proven in the previous section. 

\begin{theorem}
\label{thm_promptbound_lb}
For every $n>0$, there is a Prompt $\ltl$ formula $\phi_n'$ of size~$\bigo(n^2)$ such that
\begin{itemize}
	
	\item Player~$O$ wins $\delaygame{\phi_n'}$ for some delay function~$f$ and some $k$, but
	
	\item Player~$I$ wins $\delaygame{\phi_n'}$ for every delay function~$f$ and every $k \le 2^{2^{2^{n}}}$.

\end{itemize}
\end{theorem}

\begin{proof}
Let $\phi_n' = (\psi_0 \wedge \psi_1 ) \rightarrow (\psi_2 \wedge \psi_3 \wedge \psi_4 \wedge \psi_5 \wedge \Fp \markr)$, where the alphabets and the formulas~$\psi_i$ are as in the proof of Theorem~\ref{thm_ltllookahead_lb}. 

Let $k = f(0) = 2^n \cdot \pow{n'}$ with $n' = 2^{2^n}$ as above. Then, the strategy~$\stratO$ for Player~$O$ described in the proof of Theorem~\ref{thm_ltllookahead_lb} is winning for $\delaygame{\phi_n}$ with bound~$k$: it places both markers within the first~$f(0) = k$ positions, as it specifies a bad $j $-pair within this range.

Now, assume we have $k < 2^n \cdot 2^{n'}$, consider the strategy~$\stratI$ for Player~$I$ as defined in the proof of Theorem~\ref{thm_ltllookahead_lb}, and recall that every outcome that is consistent with $\stratI$ starts with the sequence~$w_{n'}$ in the first component.
Satisfying $\psi_2 \wedge \psi_3 \wedge \psi_4 \wedge \psi_5 $ against $\stratI$ requires Player~$O$ to mark a bad $j$-pair and to produce the sequence $j^\omega$. However, Player~$I$ does not produce a bad $j$-pair in the first $k$ positions, i.e., the conjunct~$\Fp \markr$ is not satisfied with respect to $k$. Hence, $\stratI$ is winning for Player~$I$ in $\delaygame{\phi_n}$ with bound~$k$. \qed
\end{proof}

%%%%%%%%%%%%%%%%%%%%%%%%%%%%%%%
%%%%%%%%%%%%%%%%%%%%%%%%%%%%%%%
\subsection{Lower Bounds on Complexity} 
\label{subsec_lowerbounds_cc}
Our final result settles the complexity of solving $\prompt$ delay games. The triply-exponential algorithm presented in the previous section is complemented by proving the problem to be $\threeexp$-complete, which even holds for $\ltl$. The proof is a combination of techniques developed for the lower bound on the lookahead presented above and of techniques from the $\exptime$-hardness proof for solving delay games whose winning conditions are given by deterministic safety automata~\cite{KleinZimmermann16a}.

\begin{theorem}
\label{thm_ltlcomplexity_hardness}
The following problem is $\threeexp$-complete: given an $\ltl$ formula~$\phi$, does Player~$O$ win $\delaygame{\phi}$ for some delay function~$f$?
\end{theorem}

\begin{proof}
\sloppy
Membership is proven in Theorem~\ref{thm_promptdelaycomplexity}, since $\ltl$ is a fragment of \linebreak $\prompt$. Hence, it remains to prove hardness. To this end, let $\tm = (Q, Q_\exists, Q_\forall, \Sigma, q_I, \Delta, \qacc, \qrej)$ be an alternating doubly-exponential space Turning machine with transition relation~$\Delta \subseteq Q \times \Sigma \times Q \times \Sigma \times \set{-1,+1}$ and accepting and rejecting states~$\qacc$ and $\qrej$, which we assume w.l.o.g.\ to have self-loops. Furthermore, let $p$ be a polynomial such that $2^{2^{p}}$ bounds the space-consumption of $\tm$ and let $x \in \Sigma^*$ be an input. Fix $n = p(\size{x})$.
\fussy
We construct an $\ltl$ formula~$\phi$ (of polynomial size in $n + \size{\Delta}$) such that Player~$O$ wins $\delaygame{\phi}$ for some $f$ if, and only if,  $\tm$ rejects $x$. This suffices, as $\atwoexpspace = \threeexp$ is closed under~complement.

Fix $I = \set{b_0, \ldots, b_{n-1},b_I,\sharpsym,\copyconf, \newconf } \cup \Sigma \cup Q \cup \Delta$ and $O = \set{ \errorsym, \markl, \markr } \cup \Delta$. Let $\psi_0 $ be the formula that requires Player~$I$ to implement the cyclic addressing of length~$2^n$ starting at zero using the bits~$b_j$. In the following, we only consider outcomes that satisfy $\psi_0$. Also, blocks are defined as before and we again interpret the bits~$b_I$ of a block as a number in $R = [0, 2^{2^n}-1]$. There is an $\ltl$ formula~$\theta_0$ of polynomial size that allows Player~$O$ to use the  error mark~$\errorsym$ to force Player~$I$ to implement a cyclic addressing of the blocks of length~$2^{2^n}$ starting at zero~\cite{Zimmermann13}. In particular, the formula~$\theta_0$ holds if, and only if, the first occurrence of $\errorsym$ marks a position witnessing that the addressing is not implemented correctly. Hence, Player~$O$ can satisfy $ \theta_{0} $ if, and only if, Player~$I$ incorrectly implements the cyclic addressing of the blocks of length~$2^{2^n}$. 

Assume he implements both addressings correctly: then, a superblock is an infix starting with a block encoding $0 \in R$ and that ends one position before the next block that encodes $0 \in R$, i.e., each superblock consists of $2^{2^n}$ blocks. We use such superblocks to encode configurations of $\tm$ on $x$ by placing the cell contents at the starts of the blocks.

Intuitively, Player~$I$ produces configurations of $\tm$ and is in charge of existential states, while Player~$O$ controls the universal ones and checks the configurations for correctness using the marks~$\markl$ and $\markr$ to indicate cells where the configurations were not updated correctly. To account for the lookahead in the game, which means that Player~$O$ picks her transitions to apply asynchronously, Player~$I$ is able to copy configurations in order to wait for Player~$O$'s choice. Again, Player~$O$ checks this copying process for correctness. 

To this end, we use the two propositions~$\copyconf$ and $\newconf$ to denote whether a copy or a successor configuration follows. If Player~$I$ produces a new universal configuration, then Player~$O$ has to pick some transition from $\Delta \subseteq O$, which should be applied to this configuration, possibly after some copies. If Player~$I$ copies a configuration ad infinitum, then he loses.

Consider the following assumptions on Player~$I$'s behavior (along with $\psi_0$):
\begin{enumerate}

	\item At every start of a block, exactly one proposition from $\Sigma$ holds and at most one from $Q$. Also, in each superblock, there is exactly one start of a block where a proposition from $Q$ holds. If this holds, then each superblock encodes a configuration of $\tm$ of length~$2^{2^{p(\size{x})}}$.

	\item The configuration encoded by the first superblock is the initial one of $\tm$ on $x$.

	\item At each start of a superblock, either $\copyconf$ or $\newconf$ holds, starting with $\newconf$ at the first superblock. Furthermore, we require $\newconf$ to hold infinitely often at such positions.

	\item At each start of a superblock that encodes an existential configuration, exactly one proposition from $\Delta$ holds, which has to be applicable to the configuration.

	\item There is at most one position where $\sharpsym$ holds. This is used by Player~$I$ to check Player~$O$'s error claim and is implemented as in the proof of Theorem~\ref{thm_ltllookahead_lb}.

\end{enumerate}
Each of these properties~$i \in \set{1,2,3,4,5}$ can be captured by an $\ltl$ formula~$\psi_i$ of polynomial size. Furthermore, let $\theta_1$ be an $\ltl$ formula that expresses the following, which has to be guaranteed by Player~$O$: at each start of a superblock that encodes a universal configuration, exactly one proposition from $\Delta$ holds, which has to be applicable to the configuration.

Now, we define what it means for Player~$O$ to mark an incorrectly updated configuration: the conjunction of the following properties has to be satisfied, where  we assume that $\bigwedge_{i=0}^5\psi_i \wedge \neg\theta_0 \wedge \theta_1$ holds, as this is the only case where the formula to be defined is relevant.

\begin{enumerate}
	
	\item $\markl$ and $\markr$ hold both exactly once, each at the start of a block. Furthermore, $\markr$ appears one superblock after the superblock in which $\markl$ appears.	
	
	\item If there is a position~$i_\ssharpsym$ marked by $\sharpsym$, then there are no two different positions $i$ being in the superblock of $\markl$ and $ i'$ being in the superblock of $\markr$ such that the following two conditions are satisfied: the addresses of $i$, $i'$, and $i_\ssharpsym$  are equal and $b_I$ holds at $i$ if, and only if, $b_I$ does not hold at $i'$. Such positions witness that Player~$O$ has not marked the same cell of the two subsequent configurations. This manifests itself in a single bit~$b_I$, whose address is  marked by Player~$I$.
		
	\item If $\copyconf$ holds at the start of $\markr$'s superblock, then there has to be a proposition from $\Sigma \cup Q$ that holds at the position marked by $\markl$ if, and only if, it does not hold at the position marked by $\markr$.

	\item If $\newconf$ holds at the start of $\markr$'s superblock, then let $ \delta \in \Delta$ be the unique transition holding at the last occurrence of $\newconf$ before the start of this superblock. Furthermore, let $c_m$ be the tape content (a letter from $\Sigma$ and possibly a state from $Q$) encoded at the position marked with $\markl$ and let $c_\ell$ ($c_r$) be the cell content encoded in the start of the previous (next) block. Then, $(c_\ell, c_m, c_r)$ uniquely determine the cell content of the middle cell after applying the transition~$\delta$, call it $c$. We require that the cell content encoded at the position marked by $\markr$ is different from $c$.
	
\end{enumerate}
Let $\theta_2$ be an $\ltl$ formula capturing the conjunction of these properties, which can be constructed such that it has polynomial size. 

Now, consider the $\ltl$ delay game $\delaygame{\phi}$ with $\phi = \bigwedge_{i=0}^5\psi_i \rightarrow (\theta_1 \wedge (\theta_0 \vee \theta_2 \vee \F\qrej))$. We show that $\tm$ rejects $x$ if, and only if, Player~$O$ wins $\delaygame{\phi}$ for some~$f$.

First, assume $\tm$ rejects $x$ and let $f$ be the constant delay function with $f(0) =  2\cdot 2^{2^n}$. We show that Player~$O$ wins $\delaygame{\phi}$. As long as Player~$I$ implements both addressings correctly and produces legal configurations as required by the antecedent of $\phi$,  Player~$O$ has enough lookahead to correctly claim the first error introduced by Player~$I$,  no matter whether he increments the superblock addressing incorrectly or updates a configuration incorrectly. Her strategy is to place the markers~$\errorsym, \markl, \markr$ at appropriate positions. Furthermore, she has access to the whole encoding of each universal configuration whose successor she has to determine. This allows her to simulate the rejecting run of $\tm$ on $x$, which reaches the rejecting state~$\qrej$, no matter which existential transitions Player~$I$ picks. Thus, he has to introduce an error in order to win, which Player~$O$ can detect using the lookahead. If Player~$I$ does not introduce an error, the play reaches a rejecting configuration. In every case Player~$O$ wins. 

For the converse direction, we show the contrapositive. Assume that $\tm$ accepts $x$ and let $f$ be an arbitrary delay function. We show that Player~$I$ wins $\delaygame{\phi}$. Player~$I$ implements both addressings correctly, starts with the initial configuration, and picks the successor configuration of an existential one according to the accepting run. Also, he copies universal configurations to obtain a play prefix in which Player~$O$ has to determine the transition she wants to apply in this configuration. Thus, he will eventually produce an accepting configuration without ever introducing an error. In particular, a rejecting state is never reached and Player~$O$ cannot successfully claim an error: the superblock addressing is correctly implemented and if she claims an erroneous update of the configurations, she has to mark different cells, as there are no such incorrect updates. This can be detected by Player~$I$ by placing the $\sharpsym$ at a witnessing address.  In either case, Player~$I$ wins the resulting play. \qed
\end{proof}

%%%%%%%%%%%%%%%%%%%%%%%%%%%%%%%%%%%%%%%%%%%%%%%%%%%%%%%%%%%%
%%%%%%%%%%%%%%%%%%%%%%%%%%%%%%%%%%%%%%%%%%%%%%%%%%%%%%%%%%%%
\section{Delay Games on Non-deterministic, Universal, and Alternating Automata}
\label{sec_nonaltuniv}
Finally, we argue that the lower bounds just proven for $\ltl$ delay games can be modified to solve open problems about $\omega$-regular delay games whose winning conditions are given by non-deterministic, universal, and alternating automata (note that non-determinism and universality are not dual here, as delay games are asymmetric).

\sloppy

Recall that solving delay games with winning conditions given by deterministic parity automata is $\exptime$-complete and that exponential constant lookahead is sufficient and in general necessary. These upper bounds yield doubly-exponential upper bounds on both complexity and lookahead for non-deterministic and universal parity automata via determinization, which incurs an exponential blowup. Similarly, we obtain triply-exponential upper bounds on both complexity and lookahead for alternating parity automata, as determinization incurs a doubly-exponential blowup in this case. 

\fussy

For alternating automata, these upper bounds are tight, as $\ltl$ can be translated into linearly-sized alternating automata (even with very weak acceptance conditions). Hence, the triply-exponential lower bounds proven in the previous section hold here as well.

To prove doubly-exponential lower bounds for the case of non-deterministic and universal automata, one has to modify the constructions presented in the previous section. Let us first consider the case of non-deterministic automata: to obtain a matching doubly-exponential lower bound on the necessary lookahead, we require Player~$I$  to produce an input sequence in $\set{0,1}^\omega$, where we interpret every block of $n$ bits as the binary encoding of a number in $\set{0,1, \ldots, 2^n-1}$. In order to win, Player~$O$ also has to pick an encoding of a number~$j$ with her first $n$ moves such that the sequence of numbers picked by Player~$I$ contains a bad $j$-pair. To allow the automaton to check the correctness of this pick, we require Player~$O$ to repeat the encoding of the number ad infinitum. Then, the automaton can guess and verify the two positions comprising the bad $j$-pair. Finally, to prevent Player~$O$ from incorrectly copying the encoding of $j$ (which manifests itself in a single bit), we use the same marking construction as in the previous section: Player~$I$ can mark one position~$i$ by a $\sharpsym$ to claim an error in some bit at position $i \bmod n$. The automaton can guess the value $i \bmod n$ and verify that there is no such error (and that the guess was correct). Using similar ideas one can encode an alternating exponential space Turing machine proving the $\twoexp$ lower bound on the complexity for non-deterministic automata.

For universal automata, the constructions are even simpler, since we do not need the marking of Player~$I$. Instead, we use the universality to check that Player~$O$ copies her pick~$j$ correctly. Altogether, we obtain the  results presented in Figure~\ref{fig_results}, where careful analysis shows that the lower bounds already hold for weaker acceptance conditions than parity, e.g., safety and weak parity (the case of reachability acceptance is exceptional, as such games are $\pspace$-complete for non-deterministic automata~\cite{KleinZimmermann16a}). 

\begin{figure}[h]
\centering	
	\begin{tabular}{lllll}
Automaton type &\hspace*{1em}& complexity &\hspace*{1em}& lookahead \\
\toprule

deterministic parity & & $\exptime$-complete & & exponential \\

non-deterministic parity & & $\twoexp$-complete & & doubly-exponential \\

universal parity & & $\twoexp$-complete & & doubly-exponential \\

alternating parity & & $\threeexp$-complete & & triply-exponential \\

	\end{tabular}
\caption{Overview of results for the $\omega$-regular case.}
\label{fig_results}	
\end{figure}

%%%%%%%%%%%%%%%%%%%%%%%%%%%%%%%%%%%%%%%%%%%%%%%%%%%%%%%%%%%%
%%%%%%%%%%%%%%%%%%%%%%%%%%%%%%%%%%%%%%%%%%%%%%%%%%%%%%%%%%%%
\vspace{-2.8em}
\section{Conclusion}
\label{sec_conc}
We identified $\prompt$ as the first quantitative winning condition for delay games that retains the desirable qualities of $\omega$-regular delay games: in particular, bounded lookahead is sufficient to win $\prompt$ delay games and to determine the winner of such games is $\threeexp$-complete. This complexity should be contrasted to that of delay-free $\ltl$ and $\prompt$ games, which are already $\twoexp$-complete. We complemented the complexity result by giving tight triply-exponential bounds on the necessary lookahead and on the necessary bound~$k$ for the prompt eventually operator.

All our lower bounds already hold for $\ltl$ and therefore also for (very-weak) alternating Büchi automata, since $\ltl$ can be translated into such automata of linear size~\cite{GastinO01}. On the other hand, we obtained tight matching upper bounds: solving delay games on alternating automata is $\threeexp$-complete and triply-exponential lookahead is in general necessary and always sufficient. Furthermore, our lower bounds can be modified to complete the picture in the $\omega$-regular case with regard to the branching mode of the specification automaton: solving delay games with winning conditions given by non-deterministic or universal automata is $\twoexp$-complete and doubly-exponential lookahead is sufficient and in general necessary. 

Finally, as usual for results based on the alternating-color technique, our results on $ \prompt $ hold for the stronger logics $\pltl$~\cite{DBLP:journals/tocl/AlurETP01}, $\pldl$~\cite{FaymonvilleZimmermann14}, and their variants with costs~\cite{Zimmermann15c} as well.

%%%%%%%%%%%%%%%%%%%%%%%%%%%%%%%%%%%%%%%%%%%%%%%%%%%%%%%%%%%%
%%%%%%%%%%%%%%%%%%%%%%%%%%%%%%%%%%%%%%%%%%%%%%%%%%%%%%%%%%%%

\bibliographystyle{splncs03}
\bibliography{biblio}

\newcommand{\noopsort}[1]{}
\begin{thebibliography}{10}
\providecommand{\url}[1]{\texttt{#1}}
\providecommand{\urlprefix}{URL }

\bibitem{DBLP:journals/tocl/AlurETP01}
Alur, R., Etessami, K., {La Torre}, S., Peled, D.: Parametric temporal logic
  for "model measuring". {ACM} Trans. Comput. Log.  2(3),  388--407 (2001)

\bibitem{DBLP:journals/tocl/AlurT04}
Alur, R., {La Torre}, S.: Deterministic generators and games for {LTL}
  fragments. {ACM} Trans. Comput. Log.  5(1),  1--25 (2004)

\bibitem{DBLP:conf/csl/Bojanczyk04}
Boja{\'{n}}czyk, M.: A bounding quantifier. In: Marcinkowski, J., Tarlecki, A.
  (eds.) {CSL} 2004. LNCS, vol. 3210, pp. 41--55. Springer (2004)

\bibitem{DBLP:journals/fmsd/Ehlers12}
Ehlers, R.: Symbolic bounded synthesis. Form. Method. Syst. Des.  40(2),
  232--262 (2012)

\bibitem{EmersonJutla91}
Emerson, E.A., Jutla, C.S.: Tree automata, mu-calculus and determinacy
  (extended abstract). In: FOCS 1991. pp. 368--377. IEEE (1991)

\bibitem{FaymonvilleZimmermann14}
Faymonville, P., Zimmermann, M.: Parametric linear dynamic logic. In: Peron,
  A., Piazza, C. (eds.) GandALF 2014. {EPTCS}, vol. 161, pp. 60--73 (2014)

\bibitem{DBLP:journals/fmsd/FiliotJR11}
Filiot, E., Jin, N., Raskin, J.: Antichains and compositional algorithms for
  {LTL} synthesis. Form. Method. Syst. Des.  39(3),  261--296 (2011)

\bibitem{DBLP:journals/sttt/FinkbeinerS13}
Finkbeiner, B., Schewe, S.: Bounded synthesis. {STTT}  15(5-6),  519--539
  (2013)

\bibitem{FridmanLoedingZimmermann11}
Fridman, W., L\"{o}ding, C., Zimmermann, M.: Degrees of lookahead in
  context-free infinite games. In: Bezem, M. (ed.) CSL 2011. LIPIcs, vol.~12,
  pp. 264--276. Schloss Dagstuhl - Leibniz-Zentrum für Informatik (2011)

\bibitem{GastinO01}
Gastin, P., Oddoux, D.: Fast {LTL} to {Büchi} automata translation. In: Berry,
  G., Comon, H., Finkel, A. (eds.) CAV 2001. LNCS, vol. 2102, pp. 53--65.
  Springer (2001)

\bibitem{GraedelThomasWilke02}
Gr\"adel, E., Thomas, W., Wilke, T. (eds.): Automata, Logics, and Infinite
  Games: A Guide to Current Research, LNCS, vol. 2500. Springer (2002)

\bibitem{HoltmannKaiserThomas12}
Holtmann, M., Kaiser, L., Thomas, W.: Degrees of lookahead in regular infinite
  games. LMCS  8(3) (2012)

\bibitem{DBLP:journals/ita/HornTW015}
Horn, F., Thomas, W., Wallmeier, N., Zimmermann, M.: Optimal strategy synthesis
  for request-response games. {RAIRO} - Theor. Inf. and Applic.  49(3),
  179--203 (2015)

\bibitem{HoschLandweber72}
Hosch, F.A., Landweber, L.H.: Finite delay solutions for sequential conditions.
  In: ICALP 1972. pp. 45--60 (1972)

\bibitem{KleinZimmermann15}
Klein, F., Zimmermann, M.: What are strategies in delay games? {B}orel
  determinacy for games with lookahead. In: Kreutzer, S. (ed.) CSL 2015.
  LIPIcs, vol.~41, pp. 519--533. Schloss Dagstuhl - Leibniz-Zentrum für
  Informatik (2015)

\bibitem{KleinZimmermann16a}
Klein, F., Zimmermann, M.: How much lookahead is needed to win infinite games?
  {LMCS}  {12(3)} (2016)

\bibitem{KleinZimmermann16}
Klein, F., Zimmermann, M.: Prompt delay. arXiv  1602.05045 (2016)

\bibitem{KupfermanPitermanVardi09}
Kupferman, O., Piterman, N., Vardi, M.Y.: From liveness to promptness. Form.
  Method. Syst. Des.  34(2),  83--103 (2009)

\bibitem{DBLP:conf/focs/KupfermanV05}
Kupferman, O., Vardi, M.Y.: Safraless decision procedures. In: {FOCS} 2005. pp.
  531--542. {IEEE} Computer Society (2005)

\bibitem{Mostowski91}
Mostowski, A.: Games with forbidden positions. Tech. Rep.~78, University of
  Gda\'nsk (1991)

\bibitem{Pnueli77}
Pnueli, A.: The temporal logic of programs. In: FOCS 1977. pp. 46--57. IEEE
  (1977)

\bibitem{DBLP:conf/popl/PnueliR89}
Pnueli, A., Rosner, R.: On the synthesis of a reactive module. In: POPL 1989.
  pp. 179--190. {ACM} Press (1989)

\bibitem{DBLP:conf/icalp/PnueliR89}
Pnueli, A., Rosner, R.: On the synthesis of an asynchronous reactive module.
  In: Ausiello, G., Dezani{-}Ciancaglini, M., Rocca, S.R.D. (eds.) ICALP 1989.
  LNCS, vol. 372, pp. 652--671. Springer (1989)

\bibitem{Schewe09}
Schewe, S.: Tighter bounds for the determinisation of {B}üchi automata. In:
  de~Alfaro, L. (ed.) FOSSACS 2009. LNCS, vol. 5504, pp. 167--181. Springer
  (2009)

\bibitem{Zimmermann13}
Zimmermann, M.: Optimal bounds in parametric {LTL} games. Theor. Comput. Sci.
  493,  30--45 (2013)

\bibitem{Zimmermann15c}
Zimmermann, M.: Parameterized linear temporal logics meet costs: Still not
  costlier than {LTL}. In: Esparza, J., Tronci, E. (eds.) GandALF 2015.
  {EPTCS}, vol. 193, pp. 144--157 (2015)

\bibitem{Zimmermann16}
Zimmermann, M.: Delay games with {WMSO+U} winning conditions. {RAIRO} - Theor.
  Inf. and Applic.  (2016), to appear

\end{thebibliography}

\end{document}